\documentclass{article}
\usepackage{amsmath,amsthm,amssymb,amsfonts}
\usepackage{verbatim}
\usepackage{graphicx}
\usepackage{stmaryrd} 
\usepackage{hyperref}

\usepackage[colorinlistoftodos]{todonotes}

\theoremstyle{plain}
\newtheorem{thm}{Theorem}
\newtheorem{lem}[thm]{Lemma}

\newtheorem{cor}[thm]{Corollary}

\newtheorem{remark}[thm]{Remark}
\newtheorem{alg}[thm]{Algorithm}

\theoremstyle{definition}
\newtheorem{definition}[thm]{Definition}
\newtheorem{exl}[thm]{Example}

\numberwithin{thm}{section}

\def\N{{\mathbb N}}
\def\R{{\mathbb R}}

\title{Algorithms for Dynamic Computational Geometry with Applications}

\author{Laurence Boxer
\thanks{
    Department of Computer and Information Sciences,
    Niagara University,
    Niagara University, NY 14109, USA; \newline
    and \newline
    Department of Computer Science and Engineering,
    State University of New York at Buffalo. \newline
    email: boxer@niagara.edu \newline
    ORCID: 0000-0001-7905-9643
}
}

\date{ }

\begin{document}
\maketitle{}
\begin{abstract}
Most of the literature of computational geometry concerns geometric
properties of sets of static points. M.J. Atallah introduced ``dynamic
computational geometry," concerned with both momentary and
long-term geometric properties of sets of moving point-objects. This
area of research seems to have been dormant recently. The current
paper examines new problems in dynamic computational geometry: the
``Too Close," ``Too Far," and ``3 Aligned" problems. 
Worst-case optimal solutions are given for these problems using the
sequential and coarse-grained multicomputer (CGM)
models of computation.

Key words and phrases: dynamic computational geometry, 
piecewise defined continuous function, analysis of algorithms,
coarse-grained multicomputer

\end{abstract}

\maketitle

\maketitle

\section{Introduction}
Computational geometry is generally concerned with geometric
properties of some object that is usually a set~$S$ of static points;
less often, a set of line segments, curves, or surfaces.
Examples of such properties include finding
\begin{itemize}
    \item the nearest pair of distinct members of $S$;
    \item the diameter of $S$;
    \item the convex hull of $S$;
    \item maximal subsets of (exactly or approximately) 
    collinear members of $S$
    \item subsets of $S$ that are (exactly or approximately) congruent
          to, or similar to, a given set of points,
\end{itemize}
et al. Readers needing more background in computational geometry may
find useful such sources as~\cite{PrepAndShamos,BT, deBergEtal}.

M.J. Atallah's paper~\cite{Atallah} introduced ``dynamic computational
geometry," presenting efficient sequential algorithms for a variety
of geometric problems concerned with systems of point--objects that
are in simultaneous motion in Euclidean space $\R^d$; 
it was assumed for each object that its motion was 
described in each coordinate by known polynomials of time.
Parallel versions of these algorithms were studied 
in~\cite{BM89a,BM89b,BMR98,BMR99}. Another paper
concerned with geometric properties of moving objects
is~\cite{TraEtal}. The term ``dynamic" is also applied to
the study of a system of objects in which, over time, objects are
inserted or deleted from the system, but this is not what is
studied in the current paper.

In the current paper, we introduce efficient sequential and coarse
grained parallel algorithms for additional problems in dynamic 
computational geometry (as introduced in~\cite{Atallah}), and 
discuss applications. Problems considered are:
\begin{itemize}
    \item The ``Too Close" problem: When, if ever, is the earliest
          time that one of our point objects,~$q_1$, is too close
          to any other of our point objects,~$q_j$? For moving objects,
          this seems to have safety implications.
    \item The ``Too Far" problem: When, if ever, is the earliest
          time that one of our point objects,~$q_1$, is too far from
          any other of our point objects,~$q_j$?
          This seems to have implications for the ability to
          communicate or the ability to attack.
    \item The ``3 Aligned" problem: When, if ever, is the earliest
          time that one of our point objects,~$q_1$, is (approximately
          or exactly) collinear with any two others,~$q_i$ and~$q_j$.
          This seems to have implications for 
          the ability to communicate or the ability to attack.
\end{itemize}

\section{Preliminaries}
\subsection{Motion assumptions}
\label{motionAssumps}
In the current paper, we consider a system 
$Q = \{q_i\}_{i=1}^n$ of moving point-objects in~$\R^d$. 
The motion of the point-object $q_i$ is 
described by a function $f_i: [0,M] \to \R^d$
 such that if $p_j: \R^d \to \R$ is the 
 projection to the~$j^{th}$ coordinate,
 \[ p_j(x_1, \ldots, x_n) = x_j,
 \]
 then each function  $p_j \circ f_i(t)$
 is a polynomial in~$t$.
We further assume, as
in~\cite{Atallah}, that
polynomial equations of bounded degree can be
solved in $\Theta(1)$ time. One
can imagine applications, e.g., in the safety of
airplanes circling an airport or satellites circling
the globe.

\subsection{On more general motion}
The assumption of motion $f: [0,M] \to \R^d$ such that
for each coordinate~$j$, $p_j \circ f: [0,M] \to \R$ is a polynomial in the
time variable~$t$, may seem rather artificial. However, polynomials are
dense in the set of continuous functions from~$[0,M]$ to~$\R$, i.e.,
given a continuous function $g: [0,M] \to \R$ and $\varepsilon > 0$,
there exists a polynomial $G: [0,M] \to \R$ such that for all
$t \in [0,M]$, $|g(t) - G(t)| < \varepsilon$. Further, depending on the
function~$g$, it is often possible to compute a corresponding function~$G$ efficiently, e.g.,
by using Taylor series, polynomial interpolation, et al.

Consider, for example, elliptic motion at constant angular velocity in the Euclidean plane.
We describe a point in~$\R^2$ as a pair~$(x,y)$ of
real numbers. We consider motion described by
\begin{equation}
\label{circMotion}
\begin{array}{cc}
    x = R_1 \cos(at + \theta_0) + x_0  \\
    y = R_2 \sin(at + \theta_0) + y_0
\end{array}
\end{equation}
for appropriate constants 
$R_1,R_2,a,\theta_0,x_0,y_0$, where~$t$ is
the time variable. Note that both coordinates have the
same period, and
the $x$ and $y$ coordinates describe a planar ellipse
(a circle if $R_1=R_2$), with
constant angular velocity. 

Even if a library of functions containing the $\sin$ or $\cos$ is not available,
the $\sin$ and $\cos$ functions can each be
approximated by the sums of the first several 
terms of their respective appropriately centered Taylor series.
Thus, motion as described
in~(\ref{circMotion}) is included in the set of
motions described in section~\ref{motionAssumps}. Similarly,
a wide range of motion functions not described in the form of
section~\ref{motionAssumps} can be represented in this form with
satisfactory approximation.

\subsection{Piecewise defined functions}
\begin{definition}
\label{pieceDef}
{\rm \cite{Atallah}}
Let $s \in \N$. Let $h: [0,M] \to \R$.
Let $f_1,f_2,\ldots,f_s: [0,M] \to \R$ be distinct continuous functions.
Suppose for all $t \in [0,M]$ there exists
an index~$k$ and an interval~$J=[u,v] \subset [0,M]$ such that
$t \in J$ and $h|_J = f_k|_J$.
A {piece of}~$h: [0,M] \to \R$
{\em based on} $f_1,f_2,\ldots,f_s$
is a pair~$(f_i, [u,v])$ such that
$[u,v]$ is a maximal subinterval of $[0,M]$ for which 
$h|_{[u,v]}=f_i|_{[u,v]}$ identically.
\end{definition}

\begin{definition}
    \label{switchPieceDef}
   Let $s \in \N$.
Let $f_1,f_2,\ldots,f_s: [0,M] \to \R$ be distinct continuous functions. 
    Let $h: [0,M] \to \R$ be a function defined by pieces 
    based on $f_1,f_2,\ldots,f_s: [0,M] \to \R$.
    We say $h$ {\em switches pieces at} $t_0 \in [0,M]$ if
    there are pieces of~$h$, $(f_i, [u,t_0])$ and
    $(f_j, [t_0,v])$ such that $i \neq j$, i.e., $t_0$ is a
    common endpoint of the intervals of two pieces of~$h$ that
    have distinct functions. Such a point $t_0$ is a
    {\em switchpoint} of $\{f_i\}_{i=1}^s$ on $[0,M]$.
\end{definition}

We say the set of pieces of the function~$h$ is its 
{\em description from} or {\em based on} $\{f_1, \ldots, f_s\}$, and
this set of pieces {\em describes}~$h$

\begin{remark}
    \label{switchInterior} Note that by Definition~\ref{switchPieceDef},
    a switchpoint must be an interior point of~$[0,M]$.
\end{remark}

\begin{lem}
\label{switchPtsAndPieces} A continuous function
$f: [0,M] \to \R$ has $s$ switchpoints based on $\{f_i\}_{i=1}^s$
if and only if $f$ has $s+1$ pieces based on $\{f_i\}_{i=1}^s$ on~$[0,M]$.
\end{lem}

\begin{proof}
    This follows easily from Remark~\ref{switchInterior}. Details
    are left to the reader.
\end{proof}

\subsection{Coarse grained multicomputer}
Material in this section is largely quoted or paraphrased from~\cite{BM10-11}.

The {\em coarse grained multicomputer} (CGM) was described
in~\cite{DFR93} as a model of parallel computation capable of processing
a great deal of data with relatively few processors. A $CGM(n,p)$ has~$p$
processors operating on~$\Theta(n)$ data. Every processor has~$\Omega(n/p)$
memory cells, each of~$\Theta(\log n)$ bits. ``Coarse grained" means the 
size~$\Omega(n/p)$ of local memory is ``considerable larger" than~$\Theta(1)$;
customarily, this is taken to mean $n/p \ge p$, i.e., $n \ge p^2$. We use this
convention
in the current paper. Processors may share memory or may be
arranged in some interconnection network.

Processors are indexed from 1 to $p$, and each processor ``knows" its index.

We regard a $CGM(n,p)$ as a connected graph~$G$ in which vertices are processors. If the 
processors are in an interconnection network then edges are
communication links between directly connected processors. If the processors share memory 
then we regard~$G$ as a complete graph, i.e., every pair of processors is regarded
as joined by a communication link.

Processors were assumed in~\cite{DFR93} to communicate data among themselves through
sorting operations, but we will use the assumption of later papers
that any pair of adjacent processors may exchange a unit of data 
in~$\Theta(1)$ time. Similarly, in a shared memory system, we assume any pair of
processors may exchange a unit of data in~$\Theta(1)$ time.

We will use the following.

\begin{thm}
    \label{CGMbroadcast}
    {\rm \cite{BM04}}
    A unit of data can be broadcast in a $CGM(n,p)$ from its source processor to
    all processors in~$O(p)$ time.
\end{thm}

A set of values~$S$ distributed among the processors of a parallel computer~$G$ is
said to be {\em gathered} to a processor~$P$ of~$G$ when all the values
of~$S$ are copied to~$P$.

\begin{thm}
    \label{gather-scatter}
    {\rm \cite{BM04}}
    Let $S$ be a set of $N$ elementary data items distributed among the
    processors of a $CGM(n,p)$~$G$ such that $N=\Omega(p)$ and $N=O(n/p)$.
        $S$ can be gathered to any processor of~$G$ in optimal~$\Theta(N)$ time.
\end{thm}

\begin{thm}
    \label{CGMprefix}
    {\rm \cite{BM04}}
    Let $X=\{x_i \}_{i=1}^n$ be a set of elementary comparable data items 
    distributed $n/p$ per processor in a $CGM(n,p)$. A parallel
    prefix operation on~$X$ can be computed 
    in optimal~$\Theta(n/p)$ time.
\end{thm}

\section{``Too close" problem}
Problems studied in~\cite{Atallah,BM89a,BM89b,BMR98} that may be interpreted as
concerned with safety:
\begin{itemize}
    \item The Nearest Neighbor Problem: which pair of objects are
nearest, as a function of time?
    \item The Collision Problem: when (if ever) will some pair of objects collide?
\end{itemize}
 Another important safety problem, which we call
the ``Too Close Problem": when is the earliest, if ever,
the object~$q_1$ is too close to one or more other objects~$q_j$? 

If~$d$ is a metric in which each
point-object~$P_i$ has a distance~$g_i$ such that
it is dangerous for 
\begin{equation}
\label{danger}
d(q_1,q_j) \le \min \{g_1, g_j \mid  j \neq 1 \}
\end{equation}
then we are particularly interested in predicting the time of
the first occurrence of~(\ref{danger}); perhaps the 
trajectories of one or both of the objects may be
changed to avoid such an occurrence.

If for any index $j \neq 1$,
inequality~(\ref{danger}) holds at $t=0$, then
$t=0$ is the earliest time of
a ``too close" occurrence.
Otherwise, at $t=0$, we have
$d(P_1,P_j) > \min\{q_1,q_j\}$
for all $j>1$, and since the
trajectory functions~$f_i$ are
all continuous, any instances
of inequality~(\ref{danger})
occur earliest when
$d(P_1,P_j) = \min\{q_1,q_j\}$.

\subsection{Sequential solution}
\begin{alg}
    \label{tooCloseAlg-seq}
    Let $Q$ be a set of~$n$ point-objects as in
    section~\ref{motionAssumps}. Assume for each~$i$ we know~$f_i$,
    the function that describes the motion of~$q_i$; and
    $g_i$, the minimum safe distance between~$q_i$ and any other~$q_j$.
    We solve the Too Close Problem 
for~$q_1$ as follows, where the return values of the algorithm
are in the variables $minTime$ and $earliestIndices$.
\begin{enumerate}
    \item $earliestIndices := \emptyset$ and
        $minTime := \infty$.
        \item For $j =2$ to $n$
        \begin{enumerate}
        \item Let $d_j = \min\{g_1,g_j\}$
        \item Compute the vector difference function
        \[S_j(t) := f_j(t) - f_1(t).\]
        $/*~${\rm Note $|S_j(t)|$ is the distance between
        $q_j$ and $q_1$ at time~$t$.} $*/$
        \item If $|S_j(0)| \le d_j$ then
        \begin{quote}
               minTime := 0 \\ 
               $earliestIndices := earliestIndices ~\cup~\{j\}$
        \end{quote}
        End If $|S_j(0)| \le d_j$ \\
        Else If equation~~$|S_j(t)| = d_j$~~
         has a smallest solution $t_j \in [0,M]$ then
        \begin{quote}
        If $t_j = minTime$ then 
        \begin{quote}
        $/*$ we have a tie for earliest solution found so far
        $*/$ \\
        $earliestIndices := earliestIndices ~ \cup \{j\}$
        \end{quote}
        Else If $t_j < minTime$ then
        \begin{quote}
        $/*$ we have a new earliest-so-far solution $*/$ \\
                $minTime := t_j$ \\
        $earliestIndices := \{j\}$
         \end{quote}
        End Else If $t_j < minTime$
         \end{quote}
        End Else If smallest solution $t_j$
        \end{enumerate}
        End For 
      \end{enumerate}
\end{alg}

\begin{thm}
\label{tooCloseTHm}
    Algorithm~\ref{tooCloseAlg-seq}
    executes sequentially in
    optimal~$\Theta(n)$ time.
\end{thm} 

\begin{proof}
    Clearly, step~1 executes 
    in~$\Theta(1)$ time.

    Step 2 consists of a loop
    whose body executes~$\Theta(n)$ times.
    It is elementary that each step of the loop body executes in
     $\Theta(1)$ time. Therefore step 2 requires
    $\Theta(n)$ time.

    Thus, the algorithm has a running time of $\Theta(n)$. This is 
    worst-case optimal, since in the worst case, each member of~$Q$ must be considered.
\end{proof}

\begin{remark}
    Note if all we want is the time of the first instance of too-closeness, 
    and not the index set of the too-close objects, we can change the For loop
    to a loop governed by ``While $minTime > 0$ and $j \le n$". In this case we have
    a best-case running time of~$\Theta(1)$, realized in the case $|S_2(0)| \le d_2$.
\end{remark}

\subsection{CGM solution}
We give a CGM solution that is only concerned with finding the earliest instance,
if one exists, of two members of~$Q$ being too close.

\begin{alg}
    \label{tooCloseAlg-CGM}
    Let $Q$ be a set of~$n$ point-objects as in
    section~\ref{motionAssumps}. Assume each of the $p$ processors has descriptions of
    $\Theta(n/p)$ of the pairs~$(f_i, g_i)$. Without loss of generality, 
    processor~$P_j$ has descriptions of
the records in the set~$\left \{ (f_i, g_i) \}_{i=(j-1)n/p +1}^{jn/p} \right \}$.
    We solve the Too Close Problem for~$q_1$ as follows, where the return values
    of the algorithm are in the variables $minTime$ and $earliestIndices$.
\begin{enumerate}
    \item Broadcast the pair $(f_1, g_1)$ from processor $P_1$ to all
          processors.
    \item Use Algorithm~\ref{tooCloseAlg-seq} so that each processor~$P_j$ solves its
        portion of the problem, finding the earliest time~$T_j$ that any of 
        $\{ q_i \}_{i=(j-1)n/p +1}^{jn/p}$ ($i \neq 1)$, if any, 
        is too close to~$q_1$.
    \item Gather $\{ T_j \}_{j=1}^p$ to $P_1$.
    \item Processor $P_1$ computes $m = \min\{T_j\}_{j=1}^p$ in $\Theta(p)$ time
          and notes the set~$Y$ of indices at which the minimum occurs. I.e.,
          $j \in Y$ if and only if $T_j = m$.
    \item If desired, broadcast $m$ from $P_1$ to all processors.
\end{enumerate}
\end{alg}

\begin{thm}
    Algorithm~\ref{tooCloseAlg-CGM} runs in optimal~$\Theta(n/p)$ time.
\end{thm}

\begin{proof}
\begin{enumerate}
    \item By Theorem~\ref{CGMbroadcast}, the broadcast requires~$O(p)$ time.
    \item By Theorem~\ref{tooCloseAlg-seq}, the local partial solutions require
          parallel~$\Theta(n/p)$ time.
    \item  By Theorem~\ref{gather-scatter}, the gather step requires $O(p)$ time.
    \item The minimum computed in $P_1$ requires $\Theta(p)$ sequential time.
    \item By Theorem~\ref{CGMbroadcast}, the broadcast requires~$O(p)$ time.
\end{enumerate}
     Thus the algorithm runs in~$\Theta(n/p)$ time.

The worst-case optimality of this running time follows from the 
worst-case optimality of $\Theta(n)$
as the running time of the sequential solution.
\end{proof}

\begin{remark}
    If we wish to modify Algorithm~\ref{tooCloseAlg-CGM} to return the
    index set of the Too Close objects at the earliest instance of Too-closeness,
    notice that in the worst case, all members of~$Q \setminus \{q_1\}$
    could be too close to~$q_1$ at that time, so the index set would have 
    cardinality~$n-1$, too large for the application of Theorem~\ref{gather-scatter}. We can obtain such an algorithm
    with the same asymptotic running time as Algorithm~\ref{tooCloseAlg-CGM}
    by not gathering the index set. 
\end{remark}

\section{``Too far" problem}
If our point-objects are expected to communicate with each
other, we might have an occurrence of a pair being too
far to communicate; or if~$q_1$ is required to be within a certain distance of
each of the others, we might have an occurrence in which some~$q_j$ gets too far away.
We are particularly interested in predicting the time of
the first occurrence of such an event, since the 
trajectories of one or more of the objects may be
changed to avoid such an occurrence.

This problem can be solved by algorithms similar to
Algorithms~\ref{tooCloseAlg-seq} and~\ref{tooCloseAlg-CGM}.
The sequential execution time is therefore a worst-case optimal~$\Theta(n)$, and
the execution time for a $CGM(n,p)$ is therefore a worst-case optimal~$\Theta(n/p)$.

\section{``3 aligned" problem}
Another problem that has safety implications is the
``3 aligned" problem. 
\begin{itemize}
    \item When 3 of our point-objects
are aligned (or nearly aligned), the middle one 
may interfere with communications between the outer pair.
\item If the two outer points of a (nearly) collinear triple represent 
      hostile combatants, the middle object may prevent the outer combatants from attacking each other.
\end{itemize} 

\begin{figure}
    \centering
    \includegraphics[width=1\linewidth]{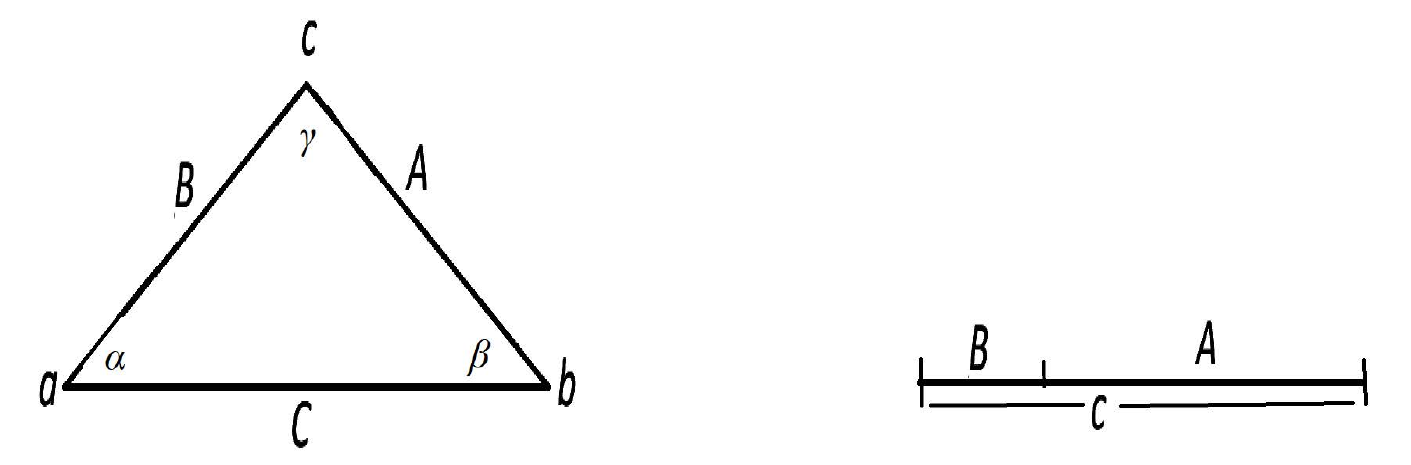}
    \caption{Assume we use the Euclidean or the Manhattan metric. \\
        Left: For a triangle as shown, with non-collinear vertices,~~~
        $C < A + B$. \\
    Right: For collinear points as shown,~~~$C = A + B$.
    }
    \label{fig:3-gon}
\end{figure}

We are particularly interested in the first predictable
occurrence of, say, $q_1$ and 2 other point-objects being (nearly) 
aligned, as perhaps we can call for at least one of them to change its
trajectory before this occurs. We give a worst-case optimal sequential solution
for this problem.

\subsection{Preliminaries}
We use the Manhattan metric as our distance function for this problem:
for $x = (x_1, \ldots, x_d),~y=(y_1, \ldots, y_d) \in \R^d$,
\[ d(x,y) = \sum_{i=1}^d |x_i - y_i|. \]

It will be useful
to determine pieces of the absolute value of a polynomial. If $f: [0,M] \to \R$ is
a polynomial function of degree at most~$s$, then $|f(t)|$ may have pieces on which
the function is~$f(t)$ and other pieces on which the function is~$-f(t)$.
We have the following.

\begin{lem}
    \label{signOfFunction}
     Let $f: [0,M] \to \R$ be continuous such that 
     for all $t \in [0,M]$, $f(t)$ can be computed in~$\Theta(1)$ time. 
      Let $[u,v]$ be a subinterval of~$[0,M]$ such that
     $u < t < v$ implies $f(t) \neq 0$.
     Let $f(t)$ and $-f(t)$ be functions used to determine pieces of $|f(t)|$.    
     Then we can determine in~$\Theta(1)$ sequential time whether
     $|f(t)|=f(t)$ for all $t \in [u,v]$,  or~$|f(t)|=-f(t)$ for all $t \in [u,v]$.
\end{lem}

\begin{proof}
In~$\Theta(1)$ time, we can pick~$t'$ such that $u < t' < v$ (e.g., $t'$ could be the
midpoint $(u+v)/2$ of~$[u,v]$). By assumption, $f(t') \neq 0$.
Since there are no zeroes of~$f(t)$ on the interval~$(u,v)$,
the continuity of~$f$ implies
\[ |f(t)| = \left \{
\begin{array}{ll}
    f(t) \mbox{ for all } t \in [u,v] & \mbox{if } f(t') > 0; \\
    -f(t) \mbox{ for all } t \in [u,v] & \mbox{if } f(t') < 0.
\end{array}  \right .
\]
The assertion follows.
\end{proof}

As an immediate consequence, we have the following.

\begin{cor}
    \label{switchPtIs0}
    Let $f: [0,M] \to \R$ be continuous. If~$\tau$ is a switchpoint
    for~$|f(t)|$ based on~$\{f,-f\}$ on~$[0,M]$, then
    $f(\tau)=0$.
\end{cor}

\begin{exl}
    \label{notAllPtsOfEquality}
    The converse of Corollary~\ref{switchPtIs0} is not in
    general true, as shown by the following.
    Let $f: [0,2] \to \R$ be given by $f(t) = (t-1)^2$. Note $f(1)=0$.
    However, since $f(t) \ge 0$ for all~$t \in [0,2]$, there is only 
    one piece of~$|f(t)|$ based on $\{f, -f \}$, namely
    the function~$f$ on the entire interval~$[0,2]$.
\end{exl}

\begin{lem}
    \label{absValueLem}
    Let $f: [0,M] \to \R$ be a polynomial function.
    \begin{itemize}
        \item If $f$ is a constant function with value~$c \ge 0$, then
              $(f, [0,M])$ is the only piece of~$|f|$ 
              based on $\{f,-f\}$, on $[0,M]$.
        \item If $f$ is a constant function with value~$c < 0$, then
              $(-f, [0,M])$ is the only piece of~$|f|$
              based on $\{f,-f\}$, on $[0,M]$.
        \item Let $s \in \N$ be bounded. If $f$ has degree $s > 0$, 
        then~$|f|$ has 
        at most~$s$ switchpoints. If~$r$ is the number of
        switchpoints of~$f$ on~$[0,M]$, then~$f$ has
        at most~$r+1$ pieces based on $\{f,-f\}$, on $[0,M]$.
    \end{itemize}
    These pieces can be computed in~$\Theta(1)$ sequential time.
\end{lem}

\begin{proof}
For a constant function $f(t)$, $|f(t)|$ clearly has only one piece, 
either $(f, [0,M])$ or $(-f, [0,M])$, respectively according as
the constant value~$c$ of~$f$ satisfies $c \ge 0$ or $c < 0$.

Suppose $f$ is not constant.
    If the equation $f(t) = 0$ has no solutions in~$[0,M]$, by 
    assumption this can be determined in~$\Theta(1)$ sequential time. 
    Then $[0,M]$ is the interval of the only piece
    of~$|f(t)|$. By Lemma~\ref{signOfFunction} in an additional~$\Theta(1)$ 
    sequential time, the function of this piece, either~$f$ or~$-f$, can be determined.

   Suppose the equation $f(t) = 0$ has solutions $t_1, \ldots, t_v$, where
    $1 \le v \le s$. Consider the following algorithm.

\begin{alg}
    \label{absValueAlg}
/* Algorithm to find pieces of $|f(t)|$ from functions
$\pm f(t)$ on $[0,M]$, where $f$ is a non-constant polynomial and has 
solutions in $[0,M]$ to the equation $f(t)=0$. */
\begin{itemize}
    \item  Solve $f(t) = 0$ for $t \in [0,M]$ and sort the unique 
 solutions $t_1, \ldots, t_v$ so
        \[ 0 \le t_1 <  \ldots < t_v \le M
        \]
        Then $1 \le v \le s$. Since we regard~$s$ as bounded, this takes $\Theta(1)$ time.
    \item If $t_1 > 0$, insert $t_0 := 0$ into this list, Similarly,
        if $t_v < M$, insert $t_{v+1} := M$ into this list. This takes $\Theta(1)$ time. The
        list now has at most $s+2$ entries, with first 
        index~$a \in \{0,1\}$ and last index~$b \in \{v,v+1\}$.
    \item  Use Lemma~\ref{signOfFunction} to determine 
                         the  function $F$ of $pieceInterval_{a}$ in~$\Theta(1)$ time.
Thus, 
\[ (pieceInterval_a, pieceFunction_a) := ([t_a,t_{a+1}], F).
\] 
\item $(previousLeft, Left) := (a, a+1)$ ~~/* prepare for next interval */
\item
\begin{tabbing} 
 For \= $right := a+2$ to $b$ \\
         \> $pieceInterval_{left} := [t_{left}, t_{right}]$ 
                         ($\Theta(1)$ time) \\
         \> Use \= Lemma~\ref{signOfFunction} to determine 
                         the function $F$ of $pieceInterval_{left}$ \\
                         \> \> in~$\Theta(1)$ time. \\
         \> If $pieceFunction_{previousLeft} = F$ then \\
             \> \> /*  same function; extend previous interval */ \\
                   \> \> $pieceInterval_{previousLeft} :=$ \\
                   \> \> ~~~~$pieceInterval_{previousLeft} \cup [t_{left}, t_{right}]$ \\
                   \> Else /* a new piece */ \\
                   \> \> $pieceFunction_{left} := F$ \\
                   \> \> $pieceInterval_{left} := [t_{left}, t_{right}]$ \\
                   \> \> $previousLeft := left$ /* prepare for next interval */ \\
                   \> End Else /* new piece */ \\
                   \> $left := right$ ~~/* prepare for next interval */ 
              \end{tabbing}
              End For        
       \item  Since we may have extended pieces to adjacent intervals, we
                    may have undefined entries $(pieceFunction_i, pieceInterval_i)$. We can compress the defined entries
                    to be consecutively indexed via a parallel
                    prefix operation 
                    in $\Theta(s) = \Theta(1)$ time, by 
                    Theorem~\ref{CGMprefix}.
    \end{itemize}
\end{alg}

    It is easily seen that the body of the For loop executes in~$\Theta(1)$ time. Therefore,
    the algorithm executes in~$\Theta(s) = \Theta(1)$ time.
\end{proof}

\begin{lem}
    \label{notPieceBd}
    Let $F,G: [0,M] \to \R$ be continuous functions. Let the members of
    $\{F_i\}_{i=1}^m \cup \{G_j\}_{j=1}^n$ 
    be polynomials of degree at most~$s$.
    Suppose~$F$ is described by pieces based on~$\{F_i\}_{i=1}^m$, 
    with~$u$ switchpoints.
    Suppose~$G$ is described by pieces based on~$\{G_j\}_{j=1}^n$,
    with~$v$ switchpoints.
    Then for each of the functions $H \in \{F + G, F-G,F \cdot G \}$,
    if $\tau$ is the endpoint of a piece of~$H$ based on 
    $\{F_i + G_j \mid 1 \le i \le m, 1 \le j \le n\}$ (respectively,
    $\{F_i - G_j \mid 1 \le i \le m, 1 \le j \le n\}$ or
    $\{F_i \cdot G_j \mid 1 \le i \le m, 1 \le j \le n\}$)
    then $\tau$ is an endpoint of a piece of a member
    of~$\{F_i\}_{i=1}^m \cup \{G_j\}_{j=1}^n$.
    Thus $H$ has at most $u+v$ switchpoints in~$[0,M]$.
\end{lem}

\begin{proof}
    We base our proof on the contrapositive: Suppose~$t_0$ is not an
    endpoint of a piece of~$F$ or of~$G$. Then we show~$t_0$ is not an
    endpoint of a piece of~$F-G$. The proof for $F+G$ and the proof 
    for~$F \cdot G$ are virtually identical.

    If $t_0$ is not an endpoint of a piece of~$F$ then the function of~$F$
    does not change its piece at~$t_0$, i.e., for some~$\varepsilon_1 > 0$,
    $F$ is identically equal to some $F_i$ on 
    $[t_0 - \varepsilon_1, t_0 + \varepsilon_1]$. 
    Similarly, if $t_0$ is not an endpoint of a piece 
    of~$G$ then the function of~$G$
    does not change its piece at~$t_0$, i.e., for some~$\varepsilon_2 > 0$,
    $G$ is identically equal to some $G_j$ on $[t_0 - \varepsilon_2, t_0 + \varepsilon_2]$. Therefore, $F-G$ is identically equal to
    $F_i - G_j$ on $[t_0 - \varepsilon_3, t_0 + \varepsilon_3]$, where
    $\varepsilon_3 = \min \{\varepsilon_1, \varepsilon_2\}$. This 
    establishes our assertion.
\end{proof}

\begin{cor}
    \label{multidim-zeroes}
    Let $f : [0,M]^d \to \R$
    where each coordinate function $p_i \circ f: [0,M] \to \R$ is a
    polynomial of degree at
    most~$s$. Then $|f|$ has at most $ds$ switchpoints, hence 
    by Lemma~\ref{notPieceBd} at most $ds+1$ pieces. The
    pieces of $|f|$ based on 
    $\{(\pm (p_1 \circ f), \ldots,  \pm (p_d \circ f)) \}$ can be computed in~$\Theta(1)$ time.
\end{cor}

\begin{proof}
Recall we use
\[ |f(t)| = \sum_{i=1}^d |p_i(f(t))|.
\]
    By Lemma~\ref{notPieceBd}, any switchpoint of~$|f|$ must be 
    a switchpoint of some~$p_i(f)$.
    The assertion follows from Lemma~\ref{absValueLem}.
\end{proof}

\begin{cor}
    \label{piecesOfDistFunct}
    Let $Q$ be a system of moving point objects as in
    section~\ref{motionAssumps}. For $i \neq j$, the distance
    function
    \[ D_{i,j}(t) = \sum_{k=1}^n |p_k(f_i(t)) - p_k(f_j(t))| 
    \]
    has at most $ds$ switchpoints, hence at most
    $ds + 1$ pieces based on the members 
    of the set of polynomial functions
    \[ \left \{ \sum_{k=1}^d \pm [p_k(f_i(t)) - p_k(f_j(t))]
       \mid 1 \le i < j \le n
       \right \}.
    \]
    These pieces may be computed in $\Theta(1)$ sequential time.
\end{cor}

\begin{proof}
    This follows from Corollary~\ref{multidim-zeroes}.
\end{proof}

\subsection{Solving the 3 aligned problem}

We are interested in the first occurrence (if any) of an
$\varepsilon$-approximately collinear triple $(q_1, q_i, q_j)$, where
$1 < i<j \le n$. As the time-dependent functions 
$A(t),B(t),C(t)$ we will use are continuous, either an 
$\varepsilon$-approximately collinear triple exists at $t=0$
or the first occurrence (if any) exists, as in
Figure~\ref{fig:3-gon}~(Right) when
\[ \varepsilon \in \{|A + B - C|,~~|A+C-B|,~~|B+C-A| \}.
\]

\begin{thm}
    \label{combine}    
    (Suggested by Lemma~2.5 of~{\rm \cite{BM89a}}.)
    Let $F, G: [0,M] \to \R$ be continuous functions.
    Let $F$ be described by pieces based on the members of 
    ${\cal F} := \{f_i\}_{i=1}^k$.
    Let $G$ be described by pieces based on the members of 
    ${\cal G} := \{g_j\}_{j=1}^m$. Suppose for each pair $(h_1,h_2)$
    of distinct members of ${\cal F} \cup {\cal G}$, the graphs
    of~$h_1$ and $h_2$ have at most~$s$ points of intersection
    in~$[0,M]$.
Then each of the functions $F + G$, $F-G$, and $F \cdot G$ has at
most $(k+m)s$ switchpoints, hence at
    most~$(k+m)s+1$ pieces based on, respectively,
    \[ \{f_i + g_j \mid f_i \in {\cal F}, g_j \in {\cal G} \},
       ~~~ \{f_i - g_j \mid f_i \in {\cal F}, g_j \in {\cal G} \},
       \mbox{ and } \{f_i \cdot g_j \mid f_i \in {\cal F}, g_j \in {\cal G} \},
    \]
    and these pieces can be computed in~$O(kms^2 + (k+m)s)$ time. If~$k,m$,
    and~$s$ are regarded as bounded, the time estimate reduces
    to~$\Theta(1)$.
\end{thm}

\begin{proof}
We give a proof for $F - G$; the proofs for the others are virtually identical.

Let ${\cal F'}$ be the ordered union of switchpoints of members of~${\cal F}$ on~$[0,M]$,
\[ t_1 < t_2 < \ldots < t_u, \mbox{ for } u \le ks. \]
Let ${\cal G'}$ be the union of switchpoints of members 
of~${\cal G}$ on~$[0,M]$,
\[ w_1 < w_2 < \ldots < w_v, \mbox{ for } v \le ms. \]
Note 
\[ \#({\cal F'} \cup {\cal G'}) \le ks + ms = (k+m)s.
\]
By Lemma~\ref{notPieceBd}, a switchpoint of $F-G$ must come from
${\cal F'} \cup {\cal G'}$, so there are at most
$(k+m)s$ switchpoints of $F-G$.

As in Lemma~\ref{absValueLem}, there are at most~$(k+m)s+1$ pieces
of~$F-G$ based 
on the members of~$\{f_i - g_j \mid f_i \in {\cal F}, g_j \in {\cal G} \}$, on~$[0,M]$.

We give our algorithm for $F-G$; the others can be handled similarly.

\begin{alg}
    \label{combineAlg}
    Compute the pieces of $F-G$ (or $F+G$ or $F \cdot G$),
    where $F$ and $G$ are functions as described above.

    If it is not known to have been done previously, sort the pieces of~$F$
    in ascending order by the left endpoints of their intervals. Since~$s$ is
    bounded, this is done in $\Theta(1)$ time.

     If it is not known to have been done previously, sort the pieces of~$G$
    in ascending order by the left endpoints of their intervals. As above, 
    this is done in $\Theta(1)$ time.

    For each pair $(\cal P, \cal Q)$ of pieces of~$F$ and of~$G$, respectively,
    use Lemma~\ref{combine} to compute the pieces of $F-G$ on
    the intersection of the interval of~$\cal P$ and the interval
    of~$\cal Q$ in~$O(1)$ time. There are at most $kms^2$ such pairs and
    at most $(k+m)s+1$ switchpoints of $F-G$. The function $F - G$
    has at most $(k+m)s+1$
    pieces based on the members of~$\{f_i - g_j \mid f_i \in {\cal F}, g_j \in {\cal G} \}$, on~$[0,M]$.
    This step executes in $O((k+m)s+1)$ time.
\end{alg}
Clearly, this algorithm runs in $O(kms^2 + (k+m)s+1)$ time.
  
If~$k$ and $m$ are regarded as bounded, the time estimate reduces
    to~$\Theta(1)$.
\end{proof}

\begin{cor}
    \label{combine3}
     Let $F, G, H: [0,M] \to \R$ be continuous functions.
    Let $F$ be described by $k$ switchpoints based on the members of 
    $\{f_i\}_{i=1}^a$.
    Let $G$ be described by $m$ switchpoints based on the members of 
    $\{g_j\}_{j=1}^b$. 
     Let $H$ be described by $r$ switchpoints based on the members of 
    $\{h_u\}_{u=1}^c$.
   Then each of the functions $\pm F \pm G \pm H$ has at
    most~$k+m+r$ switchpoints based on the members of
\[ \{ \pm f_i \pm g_j \pm h_u \mid 1 \le i \le a,~1 \le j \le b,~1 \le u \le c \}.
\] 
    If~$k,m,r,s$ are regarded as bounded, then all these switchpoints can be computed in~$\Theta(1)$ time.
\end{cor}

\begin{proof}
    This follows by applying Theorem~\ref{combine}, first to
    $F$ and $G$, then to~$F-G$ and~$H$.
\end{proof}

\begin{remark}
\label{distSwitches}
    Assume moving point-objects as in section~\ref{motionAssumps}.
    Using the Manhattan metric, $i \neq j$ implies
    \[ D_{i,j}(t) := d(f_i(t), f_j(t)) = \sum_{k=1}^d |p_k(f_i(t)) - p_k(f_j(t))|
    \]
    The polynomial function of a piece of $|p_k(f_i(t)) - p_k(f_j(t))|$ is either
    \[ p_k(f_i(t)) - p_k(f_j(t)) ~~~ \mbox{or} ~~~-[p_k(f_i(t)) - p_k(f_j(t))].
    \]
    There are at most~$s$ solutions of $p_k(f_i(t)) - p_k(f_j(t))=0$.
    Let zeroes in [0,M] of
    $p_k(f_i(t)) - p_k(f_j(t))$ be $Z_k = \{t_{k,1}. \ldots, t_{k,u_k} \}$, where
    $0 \le u_k \le s$. Then the pieces of $D_{i,j}(t)$ have at 
    most~$ds$ switchpoints in~$[0,M]$ at which some $k^{th}$ component
    switches its base function. Therefore, $D_{i,j}(t)$ has at 
    most~$ds + 1$ pieces based on $\{ \pm [p_k(f_i) -  p_k(f_j)] \}_{k=1}^d$.
\end{remark}

We are ready to consider our measure of (approximate) collinearity:
Since by the Triangle Inequality we must have
\[
    |f_i(t) - f_j(t)| \le |f_i(t) - f_1(t)| + |f_1(t) - f_j(t)|,
\] 
$q_1$ is approximately collinear with, and ``between", $q_i$ and $q_j$,
if for sufficiently small $\varepsilon > 0$,
\begin{equation}
\label{approxCollinIneq}
  |f_i(t) - f_j(t)| + \varepsilon \ge |f_i(t) - f_1(t)| + |f_1(t) - f_j(t)|
\end{equation}
Similar inequalities apply for approximate collinearity with either~$q_i$ 
``between"~$q_1$ and~$q_j$, or~$q_j$ ``between"~$q_1$ and~$q_i$.

We have that either
\begin{itemize}
    \item Approximate collinearity is satisfied at $t=0$, i.e.,
    \begin{equation}
    \label{useForTime0}
        |f_i(0) - f_j(0)| + \varepsilon \ge |f_i(0) - f_1(0)| + |f_1(0) - f_j(0)|;
    \end{equation}
     or similarly for the other ``between" possibilities, or
    \item by continuity, (\ref{approxCollinIneq}) is first satisfied for
    \begin{equation}
\label{approxCollinEq}
    \varepsilon =
    \sum_{k=1}^d |p_k[f_i(t) - f_1(t)]| +
    \sum_{k=1}^d |p_k[f_1(t) - f_J(t)]| -
    \sum_{k=1}^d |p_k[f_i(t) - f_j(t)]|
\end{equation} 
or similarly for the other ``between" possibilities.
In this case, Corollary~\ref{multidim-zeroes} and
Corollary~\ref{combine3} yield that equation ~(\ref{approxCollinEq}) 
has at most ~$3ds$ switchpoints and the pieces of the right side
of~(\ref{approxCollinEq}), based on the members of
\[ \{\pm p_k[f_i(t) - f_1(t)] \pm p_k[f_1(t) - f_k(t)]
     \pm p_k[f_1(t) - f_u(t)] \mid 1 \le k \le d \}
\]
can be computed in~$\Theta(ds)$ sequential
time. If we assume $d$ and $s$ are bounded, then the running time
reduces to~$\Theta(1)$; or
   \item (\ref{approxCollinIneq}) has no solution in $[0,M]$.
\end{itemize}

Recall the following, in which the symbol $f'$ stands for
the first derivative of~$f$.

\begin{lem}
    \label{pieceExtrema}
    Let $f: [a,b] \to \R$ be a differentiable function.
    Then each of $\min\{f(t) \mid t \in [a,b]\}$
    and $\max\{f(t) \mid t \in [a,b]\}$ occurs at some member of
    \[ \{a,b\} \cup \{t \in [a,b] \mid f'(t)=0 \}.
    \]
\end{lem}

\begin{remark}
    \label{whatsCentered}
    Note in the following Algorithm~\ref{3alignedAlg} we consider
    \[ candidates :=  \{ t \in [0,M] \mid 
        |a(t) + b(t) - c(t)| = \varepsilon \} \]
   \[    ~\cup~ \{ t \in [0,M] \mid 
        |a(t) + c(t) - b(t)| = \varepsilon \} \]
\[        \cup~\{ t \in [0,M] \mid 
        |b(t) + c(t) - a(t)| = \varepsilon \}.
    \]
    This enables us to consider instants in which any of
    $p_1(t),~p_i(t)$, or $p_j(t)$ is the middle point of a
    (near) collinear triple. If, instead, we wish to consider only
    instants when $p_1(t)$ is the middle point of a
    (near) collinear triple, we would use
        \[ candidates := \{ t \in [0,M] \mid 
        |a(t) + b(t) - c(t)| = \varepsilon \}. \]
\end{remark}

\subsection{Sequential solution}
In the following algorithm, we compute, for all index pairs $(i,j)$ such that
    $1 < i < j \le n$, $f_1(0)$, $f_i(0)$, and $f_j(0)$. This lets us
    determine from~(\ref{useForTime0}) any (if existing) solution
    occurs at $t=0$. If no such solution exists, we test for another solution. 

\begin{alg}
    \label{3alignedAlg}
    Algorithm to find first (if any) instance of an 
    $\varepsilon$-approximately collinear triple of $q_1$ and two 
    other members of~$Q$, sequentially. Assume $\varepsilon > 0$ 
    is given.

    \begin{quote}
     \begin{tabbing}
        $tFirst = \infty$, $i=2$, $j=3$ \\
        While \= $tFirst > 0$ and $i < n$ \\
        \> Com\=pute \= descriptions of functions \\
        \> \> $a(t) := d(f_1(t), f_i(t))$,~~~$b(t) :=d(f_1(t),f_j(t))$, \\
            \> \> \> $c(t) := d(f_i(t),f_j(t))$ \\
            \> If $|a(0) + b(0) - c(0)| \le \varepsilon$
            or $|a(0) + c(0) - b(0)| \le \varepsilon$ \\
            \> \> or $|b(0) + c(0) - a(0)| \le \varepsilon$ then tFirst = 0 \\
        \> Else /* min not found at $t=0$  for current $i,j$ */ \\     
        \> \> candi\=dates $:= $  \= $\{ t \in [0,M] \mid 
        |a(t) + b(t) - c(t)| = \varepsilon \}$ \\
        \> \> \> $\cup \{ t \in [0,M] \mid 
        |a(t) + c(t) - b(t)| = \varepsilon \}$ \\
        \> \> \> $\cup \{ t \in [0,M] \mid 
        |b(t) + c(t) - a(t)| = \varepsilon \}$ \\
        \> \>  If \= $candidates \neq \emptyset$ then \\
        \> \>  \> $newCandidate := \min \{t \in candidates\}$ \\
        \> \>  \> If $newCandidate < tFirst$ then $tFirst = newCandidate$ \\
        \> \> End If $candidates \neq \emptyset$ \\
        \> \> $j = j + 1$ \\
        \> \>  If $j > n$ then $i = i + 1$ and $j = i + 1$ \\
        \> End Else /* min not found at 0 */ \\
        End While \\
        Return tFirst \\     
    \end{tabbing}
    \end{quote}
\end{alg}    
    
\begin{thm}
    Algorithm \ref{3alignedAlg} runs in worst-case optimal $\Theta(n^2)$ time
    and in best-case $\Theta(1)$ time..
\end{thm}

\begin{proof}
    Since each of $a(t), b(t), c(t)$ has at most~$ds$ switchpoints,
each computation of~$candidates$ requires $\Theta(1)$ time. 
Since in the worst case there are~$\Theta(n^2)$ pairs~$(i,j)$ that must be
considered, it follows that the worst-case sequential running time of the 
algorithm is~$\Theta(n^2)$, which is optimal.

In the best case, we discover a solution at $t=0$ for $i=2$ and $j=3$, so that
the While loop body is executed just once. Thus, in this case, $\Theta(1)$ time is used.
\end{proof}

\section{Further remarks}
We have studied the Too Close, Too Far, and 3 Aligned problems in
dynamic computational geometry, for~$n$ point-objects in motion as
described in secction~\ref{motionAssumps}. We have given algorithms 
for these problems whose running times are all worst-case optimal. 
In particular, the running times are as follows.
\[
\begin{array}{lll}
    \underline{Problem} & \underline{Model~of~computation} & 
         \underline{Worst-case~time} \\
     Too~Close & Sequential & \Theta(n) \\
     Too~Close & CGM(n,p) & \Theta(n/p) \\
     Too~Far & Sequential & \Theta(n) \\
     Too~Far & CGM(n,p) & \Theta(n/p) \\
     3~Aligned & Sequential & \Theta(n^2)
\end{array}
\]

We are grateful to the anonymous referees for useful suggestions.

\section{Declarations}
\begin{itemize}
    \item Competing Interests: none
    \item Funding Information: This research did not receive any specific grant from funding agencies in the public, commercial, or not-for-profit sectors.
    \item Author contribution: The author wrote this
          entire paper.
    \item Data Availability Statement: Not Applicable
    \item Research Involving Human and/or Animals: Not
          Applicable
    \item Informed Consent: Not Applicable
\end{itemize}

\end{document}